\documentclass[pdflatex,sn-mathphys-num]{sn-jnl}
\usepackage{graphicx}
\usepackage{multirow}
\usepackage{amsmath,amssymb,amsfonts}
\usepackage{amsthm}
\usepackage{mathrsfs}
\usepackage[title]{appendix}
\usepackage{xcolor}
\usepackage{textcomp}
\usepackage{manyfoot}
\usepackage{booktabs}
\usepackage{algorithm}
\usepackage{algorithmicx}
\usepackage{algpseudocode}
\usepackage{listings}
\usepackage{comment}
\usepackage{lineno}
\usepackage{caption} 
\usepackage{geometry}          
\theoremstyle{thmstyleone}
\usepackage{enumerate}
\usepackage{enumitem}
\usepackage{titlesec}
\newtheorem{theorem}{Theorem}

\theoremstyle{thmstyletwo}

\theoremstyle{thmstylethree}
\newtheorem{definition}{Definition}
\newtheorem{lemma}{Lemma}

\theoremstyle{plain}

\newcounter{suppsubsection}

\usepackage{etoolbox}
\newtoggle{includeintoc}
\togglefalse{includeintoc}
\let\oldaddcontentsline\addcontentsline
\renewcommand{\addcontentsline}[3]{
  \iftoggle{includeintoc}{
    \oldaddcontentsline{#1}{#2}{#3}
  }{}
}

\usepackage{hyperref}

\hypersetup{
    colorlinks=true,       
    linkcolor=blue,        
    citecolor=blue,        
    urlcolor=blue,        
    pdfborder={0 0 0},     
    linkbordercolor={1 1 1} 
}
\begin{document}

\title{Exponential quantum space advantage for Shannon entropy estimation in data streams }
\author[1]{\fnm{Weijun} \sur{Feng}}
\equalcont{These authors contributed equally to this work.}

\author*[2]{\fnm{Yongzhen} \sur{Xu}}\email{xuyongzh@gmail.com}
\equalcont{These authors contributed equally to this work.}

\author*[3]{\fnm{Lvzhou} \sur{Li}}\email{lilvzh@mail.sysu.edu.cn}

\author[1]{\fnm{Gongde} \sur{Guo}}

\author*[1]{\fnm{Song} \sur{Lin}}\email{lins95@fjnu.edu.cn}

\affil[1]{\orgdiv{College of Computer and Cyber Security}, \orgname{Fujian Normal University}, \orgaddress{\city{Fuzhou}, \postcode{350117}, \state{Fujian}, \country{China}}}

\affil[2]{\orgdiv{}, \orgname{Quantum Science Center of Guangdong-Hong Kong-Macao Greater Bay Area}, \orgaddress{\city{Shenzhen}, \postcode{518045}, \state{Guangdong}, \country{China}}}

\affil[3]{\orgdiv{Institute of Quantum Computing and Software, School of Computer Science and Engineering}, \orgname{Sun Yat-sen University}, \orgaddress{\city{Guangzhou}, \postcode{510006}, \state{Guangdong}, \country{China}}}

\abstract{

Near--term quantum devices with limited qubits motivate the study of space-bounded quantum computation in the data stream model. We show that Shannon entropy estimation exhibits an exponential separation between quantum and classical space complexity in this setting. Technically, we develop a two--stage quantum streaming algorithm based on a quantum procedure with an explicitly constructed oracle derived from the streaming input. This algorithm achieves logarithmic space complexity in the accuracy parameter over the data stream, whereas any classical streaming algorithm under the same pass complexity requires polynomial space.
In sharp contrast, existing results for Shannon entropy estimation in the quantum query model achieve only a quadratic speedup.  Our work establishes a natural problem with practical applications in computer networking that admits an exponential quantum space advantage, revealing a fundamental gap between quantum query complexity and streaming space complexity.
}

\maketitle

\section{Introduction}\label{sec1}
Shannon entropy is one of the most fundamental statistics for characterizing the distribution and plays a central role in information theory \cite{shannon1948mathematical}. As an increasing amount of real-world data is generated continuously as streams rather than stored as static datasets \cite{muthukrishnan2005data,aggarwal2007data,fragkoulis2024survey}, estimating Shannon entropy in the streaming model has become a basic computational problem. 
Streaming Shannon entropy estimation plays a key role in a wide range of applications, including network anomaly detection, traffic analysis, compressed sensing and related tasks \cite{sigcommLakhinaCD05,sigcommXuZB05,zhao2007data,arackaparambil2009functional,huang2018sparse}. 
In the standard streaming model, a data stream is represented as $A=\langle x_1, x_2,\dots, x_m\rangle$ over an alphabet $[n]$, where elements arrive sequentially and are observed one at a time. Algorithms in this model may perform one or more passes over the stream while using as little space as possible. The required space is measured in bits for the classical streaming model. Over the past two decades, the space complexity of streaming Shannon entropy estimation has been extensively studied and is now well understood in the classical setting~\cite{Stacs06,sodaGuhaMV06,bhuvanagiri2006estimating,amitSODA,Harvey08,li2011new,clifford2013simple}. 

Quantum space complexity has also attracted considerable attention \cite{spaaGall06,stocGavinskyKKRW07,montanaro2016quantum,icalpHamoudiM19,focsKallaugher21,focsKallaugherP22,stocKallaugherPV24,sosaParekhKV25}. Investigating quantum space is motivated by two complementary objectives. The first is to establish provable separations between quantum and classical memory requirements. The second is to develop quantum algorithms for early fault-tolerant quantum computers with a relatively limited number of physical qubits. Recent work has shown that several streaming problems admit exponential quantum space advantages \cite{spaaGall06,stocGavinskyKKRW07,montanaro2016quantum,icalpHamoudiM19,stocKallaugherPV24}, while others provably exhibit no quantum separation \cite{focsKallaugherP22}. These contrasting results indicate that the existence of quantum space advantages depends strongly on the structure of the underlying computational problem. This naturally leads to the following fundamental question: does streaming Shannon entropy estimation admit an exponential quantum space advantage?

In this work, we answer this question affirmatively by showing that streaming Shannon entropy estimation admits an exponential quantum--classical separation in space complexity. Specifically, we develop a multi-pass quantum streaming algorithm whose space complexity has an exponentially better dependence on the accuracy parameter than classical streaming algorithms with comparable pass complexity. Previous results on Shannon entropy estimation in the quantum query model achieve only a quadratic speedup over classical algorithms~\cite{titLiW19,stocBunKT18,GilyenL20,shin2025near,chen2025list}. Our quantum algorithm is obtained through a query--to--streaming transformation by combining the classical entropy estimator~\cite{amitSODA} with the quantum Monte Carlo framework~\cite{montanaro2015quantum} using amplitude estimation~\cite{brassard2002quantum}. A key technical contribution of this work is a new quantum subroutine called Quantum Index--Position Conversion (QIPC). QIPC efficiently converts logical indices in a uniform superposition into the corresponding positions of non-majority elements in the data stream, thereby preparing the quantum state required for implementing the entropy estimator oracle using logarithmic space. Beyond Shannon entropy estimation, QIPC provides a general framework for index-to-position conversion in quantum streaming models and may serve as a useful building block for other space-efficient quantum streaming algorithms.

\section{Results}\label{sec2}

In this section, we first introduce the problem of Shannon entropy estimation in data streams, together with the quantum streaming model. We then present our main results, including quantum algorithms for this problem and corresponding classical hardness results, establishing an exponential quantum advantage. 

\subsection{Problem definition and computational model}\label{Notation and definitions}
For any integer $n> 0$, we define $ [n] = \{1, 2, \dots, n\}.$ $\log$ denotes the base-2 logarithm, and $\ln$ denotes the natural logarithm with base $e$.  For given \(\varepsilon,\delta\in(0,1)\), an algorithm is said to output an \((\varepsilon,\delta)\)-\emph{approximation} of the value $V$ if it returns a value \(\widehat V\) with probability at least \(1-\delta\)  satisfying $(1-\varepsilon)V\ \le\ \widehat V\ \le\ (1+\varepsilon)V$.

We now introduce the Shannon entropy estimation problem in the data stream setting. Let \(A=\langle x_1,x_2,\dots,x_m\rangle\) be an input stream of length \(m\), where each element \(x_j\) belongs to the set $[n] = \{1, \cdots, n\}$. 
	For every symbol \(i\in[n]\) define its frequency in the stream by
	\[
	m_i := \big|\{\, j : a_j = i \,\}\big|,
	\]
	and set the empirical probability
	\[
	p_i := \frac{m_i}{m}.
	\]
	We denote the induced empirical distribution by \(p=(p_1,\dots,p_n)\).
	The Shannon entropy of the empirical distribution
	is defined by
	\[
	H(p)\;:=\;\sum_{i=1}^n -p_i\log p_i.
	\]
Our goal is to design algorithms that compute an $(\varepsilon,\delta)$-approximation of the entropy $H(p)$ while using as few bits of memory as possible.

To further explore potential quantum advantages, we consider the quantum streaming model, which generalizes the classical framework by allowing quantum state manipulation. A quantum streaming algorithm operates in the following manner, with the objective of minimizing the number of qubits required under a fixed number of passes over the stream:
	\begin{enumerate}
		\item \textbf{Initialization.}
		The algorithm prepares an initial quantum state using a limited number of qubits, which constitute its workspace.
		
		\item \textbf{Streaming update.}
		As each element \(x_j\) arrives, the algorithm applies a unitary transformation
		\(U(x_j)\) to its current quantum state.  
		The algorithm may make one or more passes over the stream, but within each pass the elements must be processed in order as they arrive.
		
		\item \textbf{Measurement.}
		After the final pass, the algorithm performs a measurement on the quantum state to produce the output.
	\end{enumerate}

\subsection{Quantum algorithm and classical lower bound}

\begin{table*}[htbp]
\centering
\caption{Comparison of quantum and classical space complexity for Shannon entropy estimation. The $\tilde{O}$ and $\tilde{\Omega}$ notations suppress polylogarithmic factors (e.g., $\log m$, $\log n$, and $\log(1/\delta)$), emphasizing the dependence on $\varepsilon$.}
\begin{tabular}{@{}lcccc@{}}
\toprule
Problem & Passes & Classical space (bits) & Quantum space (qubits) & Separation \\
\midrule

Shannon Entropy &
$\tilde{O}(1/\varepsilon)$ &
$\tilde{\Omega}\left(\frac{1}{ \varepsilon\log^2 \varepsilon^{-1}}\right)$ &
$\tilde{O}(\log(1/\varepsilon)) $ &
Exponential \\

\addlinespace
\bottomrule
\end{tabular}
\label{tab:entropy-comparison}
\end{table*}

To rigorously demonstrate an exponential quantum advantage in Shannon streaming entropy estimation, we pursue two complementary objectives. First, we design a quantum algorithm that efficiently approximates the Shannon entropy of a data stream using minimal space. Second, we establish a classical lower bound, showing that any randomized classical algorithm requires substantially more memory to achieve the same accuracy. Together, these results quantify the separation between quantum and classical streaming approaches.  

\textbf{Quantum upper bound.} Our main quantum result, formally stated in Theorem \ref{main result}, provides an $(\varepsilon,\delta)$-approximation to the Shannon entropy using remarkably small space. For any accuracy $\varepsilon > 0$ and failure probability $\delta > 0$, the algorithm makes $\tilde{O}(\frac{1}{\varepsilon})$
passes over the data stream and uses only  $\tilde{O}(\log\frac{1}{\varepsilon})$
qubits and classical bits. 
That is, the space scales logarithmically with the accuracy $\varepsilon$.  

\textbf{Classical lower bound.} In contrast, any classical streaming algorithm that outputs an $(\varepsilon,\delta)$-approximation with $T \ge 1$ passes must use at least
\[
\Omega\Big(\frac{1}{T \varepsilon^2\log^2 \varepsilon^{-1}}\Big)
\]
bits, formally stated in Theorem \ref{lowerbound}. Unlike the logarithmic scaling of the quantum algorithm, the classical space requirement grows polynomially with $1/\varepsilon$, even for multiple passes.  

As summarized in Table \ref{tab:entropy-comparison}, our results reveal an exponential separation between classical and quantum space requirements for Shannon entropy estimation in data streams under $\tilde{O}(1/\varepsilon)$ passes. In particular, achieving a multiplicative error $\varepsilon$ requires nearly linear space in $1/\varepsilon$ for any randomized classical streaming method. In contrast, a quantum streaming algorithm can accomplish the same task using only logarithmic space in $1/\varepsilon$. Here, the $\tilde{O}$ and $\tilde{\Omega}$ notations suppress polylogarithmic factors.

\section{Discussion}\label{sec12}
The pursuit of exponential quantum--classical separations is a central goal of quantum information science, aiming to identify computational tasks where quantum algorithms can be rigorously shown to outperform classical methods. Despite significant progress, provable exponential quantum--classical separations remain relatively rare and are  known only for a limited number of computational problems~\cite{spaaGall06,stocGavinskyKKRW07,montanaro2016quantum,icalpHamoudiM19,stocKallaugherPV24,focsShor94,jordan2025optimization,Simon94,ChildsCDFGS03,li2024recovering,BenDavidCGKPW20,YamakawaZ22,babbush2023exponential,iandcLiL25a,xu2025provable,qi2026quantum,NEURIPS2024}.  Our work establishes streaming Shannon entropy estimation as another natural problem exhibiting an exponential quantum space advantage. This result highlights the potential of quantum space-efficient computation for memory-constrained data processing and motivates further investigation of quantum space advantages in streaming models.

Our work also raises several open questions. First, our results are obtained in the multi-pass streaming setting. It is natural to ask whether similar exponential separations can be achieved in the single-pass model.
Second, we focus on Shannon entropy. It remains to be understood whether R\'enyi entropy, Tsallis entropy, and other entropy generalizations also admit quantum space advantages. Third, this work focuses on vector data streams. An important direction is to investigate whether similar quantum space advantages exist for matrix streaming problems and other structured data models. 
Recent work~\cite{zhao2026exponential} has demonstrated exponential space advantages for certain matrix problems under sampling access models. However, whether similar advantages can be achieved in the streaming setting remains open.

\section{Methods}

\subsection{Quantum algorithm for Shannon entropy estimation }\label{Quantum upper Bound}

The idea behind our quantum  algorithm is to realize an important observation introduced in~\cite{amitSODA} from classical to quantum. This observation is that Shannon entropy is the expected value of a function  of $r_q$, where $r_q$ denotes the remaining repetitions of the element at position $q \in [m]$, namely, the number of times it appears in the remainder of the stream. This approach originates from the seminal work \cite{alon1996space}. 
 
For a data stream $A = \langle x_1, \ldots, x_m\rangle$, a random variable is defined based on a random suffix of the stream. Specifically, a position $q$ is selected uniformly at random from the set $\{1, \dots, m\}$. Let $r_q$ denote the number of occurrences of the element $x_q$ from position $q$ to the end of the stream, that is, 
\[ r_q = \bigl|\{\, j \in \{q, \ldots, m\} : x_j = x_q \,\}\bigr|. \] 
Define the  random variable 
\[ X(r_q) = \lambda_m(r_q)-\lambda_m(r_q-1),
\]
where $\lambda_m(r_q)=r_q \log \frac{m}{r_q}$, and $\lambda_m(0)=0$. For notational simplicity, $X(r_q)$ will hereafter be denoted by $X_q$. This construction satisfies $\mathbb{E}[X_q]=H(p)$. Therefore, estimating the Shannon entropy reduces to estimating the expectation of $X_q$. To estimate the expectation $\mathbb{E}[X_q]$, \cite{amitSODA} introduced the following classical estimation procedure.

\begin{lemma}[Lemma 2.2 in Ref.\cite{amitSODA}]\label{classcicalLemma}
    Let $X_q$ be a bounded random variable such that $-a \leq X_q \leq b$ with $a, b \geq 0$. Then there exists an algorithm that uses $O\left(\frac{\log(1/\delta)(a+b)(\mathbb{E}[X_q]+a)}{\varepsilon^{2}\mathbb{E}[X_q]^2}\right)$ independent samples of $X_q$ and outputs an $(\varepsilon,\delta)$-approximation of $\mathbb{E}[X_q]$.  The total space complexity is obtained by multiplying the sample complexity by the space required to store and process each sample.
\end{lemma}

Another important observation is that the entropy is influenced by the frequency of the majority element, which can be seen directly from the entropy structure.
Let $p_x$ denote the probability associated with an arbitrary element $x$. From the definition of the Shannon entropy, it follows that the entropy can be decomposed as
\[
H(p)
=
p_{x}\log\frac{1}{p_{x}}
+
\sum_{i\ne x} p_i \log \frac{1}{p_i}.
\]
As $p_x\rightarrow 1$, the first term vanishes since $\log(1/p_{x}) \to 0$, while the second term is suppressed by the remaining $1-p_{x}$. Consequently, a large frequency of the majority element $m_x$ leads to a small entropy and, in turn, increases the complexity implied by Lemma \ref{classcicalLemma}.

To address this problem, Refs. \cite{amitSODA,Stacs06} partition the analysis into two cases according to whether a majority element exists: $m_x\le m/2$ and $m_x>m/2$. When no majority element exists ($m_x\le m/2$), the random variable $X_q$ satisfies $H(p)= \mathbb{E}[X_q]$, which can be estimated using Lemma \ref{classcicalLemma}. When $m_x>m/2$, the algorithm removes the majority element $x$ and defines a new random variable $X'_q$ over the restricted substream $A\backslash \{x\}$ in the same manner as $X_q$. The expectation of $X'_q$, together with the contribution of the majority element, yields $H(p)=\frac{m-m_x}{m}\mathbb{E}[X'_q]+\frac{m_x}{m}\log\frac{m}{m_x}$, where $\mathbb{E}[X'_q]$ is estimated using Lemma \ref{classcicalLemma}.

\subsubsection{Quantum query algorithm with an implementable oracle}
Inspired by the classical estimation procedure established in Lemma \ref{classcicalLemma}, we present a quantum query algorithm for estimating the Shannon entropy in the streaming setting. Suppose that we are given access to a quantum oracle that outputs $X_q$. This is formalized as follows. 
\begin{definition}[Quantum Oracles]\label{oracle}
Given a position $q \in \{1, \dots, m\}$, there exists a quantum oracle $O$ such that
\[
|q\rangle |z\rangle \xrightarrow{O} |q\rangle |z+X_q\rangle
\]
acting on $O(\log m + \log n)$ qubits of space, corresponding to the encoding of the stream position and the associated estimator value.
\end{definition}
With access to this oracle, we can now design a quantum algorithm to estimate $\mathbb{E}[X_q]$ efficiently, as stated in Lemma \ref{Expectation}.

\begin{lemma}[Quantum query algorithm.]\label{Expectation}
Let $X_q$ be a random variable satisfying $-a \le X_q \le b$ for some $a,b \ge 0$, and assume oracle access to $X_q$ as given in Definition \ref{oracle}. Then there exists a quantum query algorithm that outputs an \((\varepsilon,\delta)\)-{approximation} to $\mathbb{E}[X_q]$
using
\[
O\!\left(
\frac{\log(1/\delta) \sqrt{(a+b)(\mathbb{E}[X_q]+a)}}{\varepsilon \mathbb{E}[X_q]}
\right)
\]
queries, and requiring
\[
O\!\left(\log m + \log n + \log \frac{a+b}{\varepsilon \mathbb{E}[X_q]}\right)
\]
qubits of space.
\end{lemma}
\begin{proof}
The quantum query algorithm is obtained by applying quantum Monte Carlo methods~\cite{montanaro2015quantum} to the entropy estimator introduced in the classical streaming algorithm~\cite{amitSODA}.
	
	We first define a new random variable $Y_q=(X_q+a)/(a+b)\in[0,1]$, and let $\nu=\mathbb{E}[Y_q]=(H(p)+a)/(a+b)$. A relative estimation of $\nu$ allows the recovery of the entropy $H(p)$. The procedure proceeds as follows. Assume that the superposition state
    \begin{equation}
        \sum_{q\in\mathcal{Q}}\sqrt{p_q}|q\rangle
    \end{equation}
    associated with a distribution $\{p_q\}_{q\in\mathcal{Q}}$ is provided in advance, where $\mathcal{Q}$ represents the set of all possible values of $q$. To coherently encode the random variable $X_q$, we first apply the oracle $O$ to the state $\sum_{q\in\mathcal{Q}}\sqrt{p_q}|q\rangle|0\rangle$, obtaining
	\begin{equation}
		O\sum_{q\in\mathcal{Q}}\sqrt{p_q}|q\rangle|0\rangle=\sum_{q\in\mathcal{Q}}\sqrt{p_q}|q\rangle|X_q\rangle,
	\end{equation}
	where $p_q$ denotes the probability associated with position $q$. Next, we apply quantum arithmetic to transform \(X_q\in[-a,b]\) into \(Y_q\in[0,1]\), obtaining
    \begin{equation}
        \sum_{q\in\mathcal{Q}}\sqrt{p_q}|q\rangle|X_q\rangle|Y_q\rangle.
	\end{equation}
    We further apply a controlled $R_y$ rotation conditioned on $|Y_q\rangle$, $R_y(2\arcsin{\sqrt{Y_q}})$, yielding the state 
	\begin{equation}
		\sum_{q\in\mathcal{Q}}\sqrt{p_q}|q\rangle|X_q\rangle|Y_q\rangle\left(\sqrt{Y_q}|1\rangle+\sqrt{1-Y_q}|0\rangle\right).
	\end{equation}
    This defines a unitary operator $U$ satisfying
	\begin{equation}
		U|0\rangle=\sum_{q\in\mathcal{Q}}\sqrt{p_q\cdot Y_q}|q\rangle|X_q\rangle|Y_q\rangle|1\rangle+\sum_{q\in\mathcal{Q}}\sqrt{p_q\cdot (1-Y_q)}|q\rangle|X_q\rangle|Y_q\rangle|0\rangle,
	\end{equation}
	where the probability of measuring ``1'' in the last qubit is
	\begin{equation}
		\Pr\left(1\right)=\sum_{q\in\mathcal{Q}}p_q\cdot Y_q=\mathbb{E}[Y_q]=\nu.
	\end{equation}
    Given the unitary \(U\), the standard amplitude estimation algorithm \cite{brassard2002quantum} provides an estimate \(\tilde{\nu}\) satisfying
\begin{equation}\label{amp_est-err}
    |\tilde{\nu}-\nu|\le C\left(\frac{\sqrt{\nu}}{t}+\frac{1}{t^2}\right),
\end{equation}
where \(t\) is the accuracy parameter and \(C\) is a universal constant. The estimate succeeds with probability at least \(8/\pi^2\), which can be boosted to \(1-\delta\) by \(O(\log(1/\delta))\) repetitions using standard median amplification. Note that, by the definition of $\nu$, estimating $\nu$ allows us to recover the entropy via
    \begin{equation}
        H(p)=(a+b)\nu-a.
    \end{equation}
    We next show how to obtain an $(\tilde{\varepsilon},\delta)$-approximation of $\nu$, which further leads to an $(\varepsilon,\delta)$-approximation of $H(p)$.

To achieve the relative error bound $|\tilde{\nu}-\nu|\le\tilde{\varepsilon}\nu$, Eq.(\ref{amp_est-err}) implies that it suffices to take
\[
t=O\left(
\frac{1}{\tilde{\varepsilon}\sqrt{\nu}}
\right).
\]
We next determine the relationship between \(\tilde{\varepsilon}\) and \(\varepsilon\) so that an \((\tilde{\varepsilon},\delta)\)-approximation of \(\nu\) directly yields an \((\varepsilon,\delta)\)-approximation of \(H(p)\). Define the entropy estimator as 
\[
\tilde{H}(p):=(a+b)\tilde{\nu}-a.
\]
Then,
\begin{equation}
    |\tilde{H}(p)-H(p)|=(a+b)|\tilde{\nu}-\nu|\le (a+b)\tilde{\varepsilon}\nu.
\end{equation}
To ensure $|\tilde{H}(p)-H(p)|\le \varepsilon H(p)$, it suffices to choose
\begin{equation}
    \tilde{\varepsilon}=\frac{\varepsilon\cdot H(p)}{(a+b)\nu}=\frac{\varepsilon\cdot H(p)}{H(p)+a},
\end{equation}
where $\nu=(H(p)+a)/(a+b)$. Then, with $\mathbb{E}[X_q]=H(p)$, and substituting this into the expression for $t$ yields
\begin{equation} 
t=O\Bigg(\frac{\sqrt{(a+b)(\mathbb{E}[X_q]+a)}}{\varepsilon\mathbb{E}[X_q]}\Bigg). \end{equation}
Consequently, an \(\left(\frac{\varepsilon\cdot \mathbb{E}[X_q]}{\mathbb{E}[X_q]+a},\delta\right)\)-approximation of \(\nu\) induces an \((\varepsilon,\delta)\)-approximation of \(H(p)\):
\begin{equation} |\tilde{H}(p)-H(p)|=(a+b)|\tilde\nu-\nu|\le(a+b)\tilde{\varepsilon}\nu= \varepsilon H(p). 
\end{equation}

Finally, we analyze the resource requirements of the algorithm. 

\textbf{Space and query complexity analysis.} We first consider the space complexity of state preparation. Preparing the state $\sum_{q\in\mathcal{Q}}\sqrt{p_q}|q\rangle$ requires $O(\log m)$ qubits. Note that we only consider the space required to represent $\sum_{q\in\mathcal{Q}}\sqrt{p_q}|q\rangle$ here; the auxiliary workspace required for its preparation will be analyzed later when the corresponding algorithm is invoked. Applying the oracle \(O\) to obtain $\sum_{q\in\mathcal{Q}}\sqrt{p_q}|q\rangle|X_q\rangle$ introduces an additional workspace of $O(\log m+\log n)$ qubits, as established in Lemma \ref{thm:QS-RVEE}. The quantum arithmetic operation computing \(Y_q=(a+X_q)/(a+b)\) requires \(O(\log(a+b))\) qubits for fixed-point representation, where the values of $a=\log e$ and $b=\log m$ are fixed according to \cite{amitSODA}. In addition, the controlled rotation $R_y(2\arcsin{\sqrt{Y_q}})$ uses one ancilla qubit. The amplitude estimation procedure introduces an additional register of size \(O(\log t)\), where $t = O\Bigg(\frac{\sqrt{(a+b)(\mathbb{E}[X_q]+a)}}{\varepsilon\mathbb{E}[X_q]}\Bigg)$.
Therefore, the overall space complexity is
\begin{equation}
O\left(
\log m+\log n+\log \frac{a+b}{\varepsilon\mathbb{E}[X_q]}\right).
\end{equation}

We now consider the query complexity of the algorithm. The standard amplitude estimation algorithm outputs \(\tilde{\nu}\) with probability \(8/\pi^2\) using \(t\) oracle calls. To amplify the success probability to \(1-\delta\), the procedure is repeated \(O(\log(1/\delta))\) times. Therefore, the total number of oracle \(O\) calls is
\begin{equation}
O\!\left(
t\log\frac{1}{\delta}
\right)
=O\left(
\frac{\log(1/\delta) \sqrt{(a+b)(\mathbb{E}[X_q]+a)}}{\varepsilon \mathbb{E}[X_q]}
\right).
\end{equation}
This completes the proof of Lemma \ref{Expectation}.
\end{proof}
Having established the query algorithm, the next step is to convert it into a quantum streaming algorithm. This requires two key considerations: first, constructing the oracle in the streaming setting, and second, ensuring that both the space complexity and the number of passes over the stream meet the desired quantum bounds, thereby exceeding classical methods. 
The construction of the oracle is addressed in Lemma \ref{thm:QS-RVEE}.
\begin{lemma}[Oracles $O$]\label{thm:QS-RVEE}
	For any $q \in [m]$ with each $x_q \in [n]$ in the stream $A$, there exists a quantum circuit to implementing $|q\rangle |z\rangle \xrightarrow{O} |q\rangle |z+X_q\rangle$, using $O(\log m + \log n)$ qubits of space in two passes over the stream. 
\end{lemma}

 \begin{proof}
To prove this lemma, we first present the quantum streaming algorithm for constructing the oracle $O$. The construction of oracle is based on the fundamental principle that classical computational circuits can be efficiently transformed into quantum circuits \cite{bennett1973logical,toffoli1980reversible,siamcompBennett89,siamcompLevinS90}. This principle was first exploited in the quantum streaming model by  Ref.~\cite{montanaro2016quantum}. We then analyze the cost of the oracle, including its space complexity and the number of passes over the data stream.
\begin{algorithm}[h]
	\caption{Quantum streaming algorithm for constructing the oracle $O$}
	\label{Algorithm1}
	\begin{algorithmic}[1]
		\Statex \textbf{Input:} data stream $A=\langle x_1,x_2,\dots,x_m\rangle$ with $x_j\in [n]$, $q\in [m]$, and $z\in\mathbb{R}$
		\State Initialize the state $|q\rangle|0\rangle|z\rangle$.
		\Statex \textbf{First pass:}
		\State Set $r_q\gets 0$;
		\Statex \For{each update $x_j$ from $x_1$ to $x_m$}
		{
			\Statex $U_{x_j} \colon | q \rangle | r_q \rangle \mapsto | q \rangle \big| r_q + g_q (x_j) \big\rangle$
			\Statex $r_q\gets r_q +g_q (x_j)$
		}
		\EndFor
		\State Prepare the state $|q\rangle|r_q\rangle|z+X_q\rangle$, where $X_q=r_q\log\frac{m}{r_q}-(r_q-1)\log\frac{m}{(r_q-1)}$ is computed from $r_q$ and $m$ using quantum arithmetic operations.
		\Statex \textbf{Second pass:}
		\State Uncompute the register $|r_q\rangle$, where \Statex\For{each update $x_j$ from $x_1$ to $x_m$}{\Statex $U_{x_j}^{-1} \colon | q \rangle | r_q \rangle \mapsto | q \rangle \big| r_q - g_q (x_j) \big\rangle$
			\Statex $r_q\gets r_q - g_q (x_j)$.}
		\EndFor
		\Statex \textbf{Output:} $|q\rangle|z+X_q\rangle$.
	\end{algorithmic} 
\end{algorithm} 

	\textbf{Quantum streaming algorithm for constructing the oracle $O$.} The proposed algorithm, Algorithm \ref{Algorithm1}, implements a quantum streaming procedure to construct the oracle $O$, which outputs $X_q$ for a given position $q$ after two passes over the data stream. The circuit realization is depicted in Fig. \ref{Fig1}.
	\begin{figure}[h]
		\centering
		\includegraphics[width=0.56\textwidth]{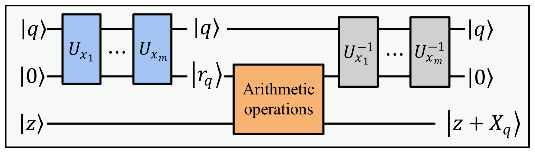}
		\caption{The circuit for implementing the oracle $O$. Here, each unitary operator $U_{x_j}$ and its inverse $U_{x_j}^{-1}$, for $j\in[m]$, are induced by the elements of the data stream. Arithmetic operations can be viewed as a sequence of quantum operators that implement the function $X_q=r_q\log \frac{m}{r_q}-(r_q-1)\log \frac{m}{(r_q-1)}$.}
		\label{Fig1}
	\end{figure}
	Let $x_q$ be the element at position $q$ and let $\mathcal{G}$, $|\mathcal{G}|=m$, be a family of functions such that for each reversibly implemented mapping $g:(q, x_j)\mapsto g_q(x_j)$, the function $g_q (x_j)$ is 1 if $j\ge q$ and $x_j=x_q$, and 0 otherwise, using space $\log m+\log n$, as derived in the subsequent complexity analysis. For each upcoming element $x_j$, the map
	\begin{equation}
		U_{x_j} \colon |q\rangle|y\rangle \mapsto | q \rangle \big|y+g_q (x_j)\big\rangle
	\end{equation}
	can be implemented using the function $g_q(\cdot)$. After the entire stream has been processed, we sequentially apply the mappings $U_{x_j}$ to obtain
	\begin{equation}\label{map1}
		U_{x_m} \cdots U_{x_1} |q\rangle|0\rangle \mapsto |q\rangle\left|\sum_{j\ge q}g_q (x_j)\right\rangle,
	\end{equation}
	where $\sum_{j\ge q}g_q (x_j)$ is the number of occurrences of element $x_q$ in the sub-stream $\langle x_q,\dots,x_m\rangle$, denoted as $r_q$. From $r_q$, one can construct the state
	\begin{equation}
		|q\rangle|r_q\rangle|z+X_q\rangle,
	\end{equation}
	where $z$ is an arbitrarily initialized real number, and $X_q$ is computed using quantum arithmetic operations. To decouple ancillary registers, we uncompute the second register by re-reading the stream and applying a series of maps defined as:
	\begin{equation}\label{map2}
		U_{x_j}^{-1}:|q\rangle|y\rangle\mapsto|q\rangle\left|y-g_q(x_j)\right\rangle.
	\end{equation}
	Once the second pass over the stream is completed, we sequentially apply the maps $U_{x_j}^{-1}$ to obtain
	\begin{equation}
		U_{x_m}^{-1} \cdots U_{x_1}^{-1} |q\rangle|r_q\rangle|z+X_q)\rangle \mapsto |q\rangle\left|0\right\rangle|z+X_q\rangle.
	\end{equation}
	The inverse computation here uses $U_{x_m}^{-1} \cdots U_{x_1}^{-1}$ instead of $U_{x_1}^{-1} \cdots U_{x_m}^{-1}$ to ensure that the second pass of the data stream loads the stream elements in the same order as the first pass. Note that its implementation is possible because the operations induced by any permutation of the suffix $\langle x_q,\cdots, x_m\rangle$ can realize the inverse operation of $r_q=\sum_{j\ge q}g_q(x_j)$. In addition, the final state $|q\rangle|z+X_q\rangle$ encodes $X_q$ associated with an position $q$.

\textbf{Space and pass complexity analysis.} We now specify the space cost and the number of passes over the data stream required to implement the oracle. The space complexity analysis follows from describing the state evolution of Algorithm \ref{Algorithm1}. The input state is $|q\rangle|0\rangle|z\rangle$, where the register $|q\rangle$ requires $\log m$ qubits and, for simplicity, $z$ is treated as a constant. Applying the operator sequence $\prod_{j=1}^{m} U_{x_j}$ to the first two registers yields the state $|q\rangle|r_q\rangle|z\rangle$, where $r_q \in [1,m]$ and thus also requires $\log m$ qubits. Each unitary $U_{x_j}$ is implemented by a reversible circuit using $q$, $j$, $x_q$, and $x_j$ as control information. The construction consists of standard reversible subroutines, including a comparator for testing $j \ge q$, an equality-check circuit for $x_j$ and $x_q$, and a controlled update operation. These components can be realized without measurement, using $O(\log m + \log n)$ ancilla qubits. Since the ancillary workspace is reused and uncomputed after each application of $U_{x_j}$, the sequential implementation of $\prod_{j=1}^{m} U_{x_j}$ incurs no additional asymptotic space cost.

A subsequent arithmetic circuit updates the third register, producing $|q\rangle|r_q\rangle|z + X_q\rangle$. To represent the random variable $X_q$ in the last quantum register, we first determine the required precision; since for any position $q$ it holds that $X_q < m\log m$, it follows that representing $X_q$ requires at most $\log m$ qubits, and in addition the arithmetic circuit used to compute $X_q$ introduces an extra ancilla space overhead of $O(\log m)$ to ensure the desired accuracy in representing $X_q$. Hence, the total space complexity for preparing $|q\rangle|r_q\rangle|z + X_q\rangle$ is $O(\log m + \log n)$.

Finally, applying $\prod_{j=1}^{m} U_{x_j}^{-1}$ restores the second register, yielding $|q\rangle|0\rangle|z + X_q\rangle$. This inverse sequence has the same space cost as the forward one. Therefore, the overall space complexity of Algorithm \ref{Algorithm1} is $O(\log m+\log n)$ qubits.

The pass complexity analysis is straightforward. Obtaining the state $|q\rangle|0\rangle|z+X_q\rangle$ requires applying the two operator sequences $\prod_{j=1}^{m} U_{x_j}$ and $\prod_{j=1}^{m} U_{x_j}^{-1}$. Each such sequence is induced by a single pass over the data stream; therefore, each invocation of the oracle $O$ requires two passes over the stream. 

For completeness, we remark that \(O^{-1}\), as the inverse operation of $O$, can be implemented in essentially the same manner as $O$. Hence, its space and pass complexities follow immediately from those of $O$, and we omit further discussion. 

This completes the proof of Lemma \ref{thm:QS-RVEE}.
\end{proof}

\subsubsection{Two-stage quantum streaming algorithm}

Here, we design a two-stage quantum streaming algorithm that treats streams with and without a majority element separately, dividing the estimation procedure into two cases. In both cases, estimating the respective expectations $\mathbb{E}[X_q]$ and $\mathbb{E}[X'_q]$ relies on Lemma \ref{Expectation}, whose application requires constructing the distribution over stream positions for querying the oracle $O$. When $m_x\le m/2$, the position state $\frac{1}{\sqrt{m}}\sum_{q}|q\rangle$ can be prepared directly. In contrast, when $m_x>m/2$, constructing the position state $\frac{1}{\sqrt{m-m_x}}\sum_{q:x_q\ne x}|q\rangle$ is no longer straightforward. To address this issue, we present the following lemma.

\begin{lemma}[Quantum Index-Position Conversion, QIPC]\label{thm:QIPC}
	Given a data stream $A=\langle x_1,\cdots,x_m\rangle$ of length $m$, and the frequency $m_x$ of any element $x\in [n]$, there exists a one-pass quantum streaming algorithm that prepares the state $\frac{1}{\sqrt{m-m_x}}\sum_{q:x_q\ne x}|q\rangle$, using $O(\log m + \log n)$ qubits of space.
\end{lemma}  
\begin{proof}
     To prove this result, we first present the quantum subroutine QIPC. We then analyze its implementation cost, including the space complexity and the number of passes over the data stream.

\textbf{Implementation of the quantum subroutine QIPC.} To illustrate the implementation of QIPC, we first present an example to convey the intuition. We then provide a general construction of the position distribution over all stream elements except $x$.

Consider the data stream $A=\langle x,a,x,b,x\rangle$. For clarity, we describe the evolution from an easily preparable uniform superposition (rather than from $|0\rangle$), where the procedure transforms \(\frac{1}{\sqrt{2}}(|1\rangle+|2\rangle)\) into \(\frac{1}{\sqrt{2}}(|2\rangle+|4\rangle)\). We view the evolution as a sequential update of the quantum state along the stream. Each update is driven by the incoming element $x_j$: if $x_j=x$, a shift operation is applied to all unfixed positions; otherwise, one unfixed position is fixed, leaving the quantum state unchanged. This induces a deterministic state-transition process along the stream.

For the above example, the evolution can be summarized as follows:
\begin{equation}
    \frac{1}{\sqrt{2}}\left(|1\rangle+|2\rangle\right)\overset{x}{\longrightarrow}\frac{1}{\sqrt{2}}(|2\rangle+|3\rangle)\overset{a}{\longrightarrow}\frac{1}{\sqrt{2}}(\textcolor{red}{|2\rangle}+|3\rangle)\overset{x}{\longrightarrow}\frac{1}{\sqrt{2}}(\textcolor{red}{|2\rangle}+|4\rangle)\overset{b}{\longrightarrow}\frac{1}{\sqrt{2}}(\textcolor{red}{|2\rangle}+\textcolor{red}{|4\rangle})\overset{x}{\longrightarrow}\frac{1}{\sqrt{2}}(\textcolor{red}{|2\rangle}+\textcolor{red}{|4\rangle}).
\end{equation}
The key idea of QIPC is that the state is updated in an online manner, where each stream element either triggers a shift of all unfixed positions or fixes one position corresponding to a non-$x$ element.

We now move beyond the simple case of QIPC and consider the implementation of this quantum subroutine. Specifically, we prepare a quantum state $\left|\psi\right\rangle=\frac{1}{\sqrt{m-m_x}}\sum_{q:x_q\ne x}|q\rangle$ without prior knowledge of the indices $q$ satisfying $x_q\ne x$. The preparation procedure consists of three stages, as illustrated in Fig. \ref{Fig2}.
\begin{figure}[t]
	\centering
	\includegraphics[width=0.56\textwidth]{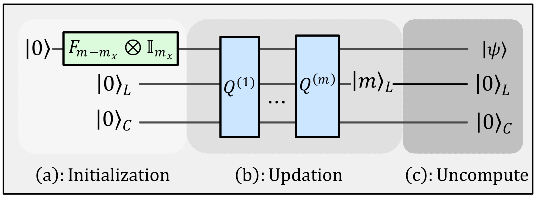}
	\caption{{The circuit for implementing the QIPC. Here, $F_{m-m_x}\otimes \mathbb{I}_{m_x}$} represents a quantum operation that prepares a uniform superposition state only in the first $m-m_x$ dimensions, while the operations from $Q_1$ to $Q_m$ represent unitary transformations induced by elements of the data stream.}
	\label{Fig2}
\end{figure}
\textbf{(a) Initialization Operation:} In the initialization stage, the quantum state is prepared as
\begin{equation}
	\left|\phi^{(0)}\right\rangle=\frac{1}{\sqrt{m-m_x}}\sum_{q'=1}^{m-m_x}|q'\rangle_I|0\rangle_L|0\rangle_{C},
\end{equation}
by applying the operator $F_{m-m_x}\otimes\mathbb{I}_{m_x}$ to the index register $I$, initialized in the state $|0\rangle$. Here, the subscript $I$ denotes an index register consisting of $\log m$ qubits, $L$ denotes a counter register storing the stream length, and $C$ denotes an auxiliary flag register.\\
\textbf{(b) Update Operation:} In this stage, we describe the evolution of the quantum state recursively. Starting from the initial state \(\left|\phi^{(0)}\right\rangle\), the system evolves under a sequence of unitary operators $\{Q^{(j)}\}_{j=1}^{m}$ induced by the data stream, where for each $j\in [m]$,
\begin{equation}
    Q^{(j)}\left|\phi^{(j-1)}\right\rangle=\left|\phi^{(j)}\right\rangle.
\end{equation}
To analyze the structure of the state $\left|\phi^{(j)}\right\rangle$ for any $j\in [m]$, we assume that
\begin{equation}
    \left|\phi^{(j-1)}\right\rangle=\frac{1}{\sqrt{m-m_x}}\left(\sum_{q'=j-\sum_{k=1}^{j-1}\mathbf{1}_{x_k = x}}^{m-m_x}|q'+\sum_{k=1}^{j-1}\mathbf{1}_{x_k = x}\rangle_I+\sum_{k=1}^{j-1}(1-\mathbf{1}_{x_k = x})|k\rangle_I\right)|j-1\rangle_L\left|0\right\rangle_C
\end{equation}
admits a tractable decomposition over the computational registers. Under this inductive hypothesis, applying $Q^{(j)}$ yields an explicit form of $\left|\phi^{(j)}\right\rangle$, preserving the same structural decomposition. Specifically, the arrival of element \(x_j\) triggers the operation \(Q^{(j)}\), the \(L\) register is incremented by one, the \(C\) register stores the indicator \(\mathbf{1}_{x_j = x}\), resulting in an intermediate state
\begin{equation}
    \frac{1}{\sqrt{m-m_x}}\left(\sum_{q'=j-\sum_{k=1}^{j-1}\mathbf{1}_{x_k = x}}^{m-m_x}|q'+\sum_{k=1}^{j-1}\mathbf{1}_{x_k = x}\rangle_I+\sum_{k=1}^{j-1}(1-\mathbf{1}_{x_k = x})|k\rangle_I\right)|j\rangle_L\left|\mathbf{1}_{x_j = x}\right\rangle_C.
\end{equation}
We then apply the operator $|0\rangle\langle 0|\otimes \mathbb{I}+|1\rangle\langle 1|\otimes U_j$ to register $I$, controlled by register $C$, and obtain the state
\begin{equation}
    \frac{1}{\sqrt{m-m_x}}\left(\sum_{q'=j+1-\sum_{k=1}^{j}\mathbf{1}_{x_k = x}}^{m-m_x}|q'+\sum_{k=1}^{j}\mathbf{1}_{x_k = x}\rangle_I+\sum_{k=1}^{j}(1-\mathbf{1}_{x_k = x})|k\rangle_I\right)|j\rangle_L\left|\mathbf{1}_{x_j = x}\right\rangle_C,
\end{equation}
where the unitary operator $U_j$ is divided into two cases: when $j=1$, it is given by
\begin{equation}
    U_1=\sum_{q=1}^{m-1}|q+1\rangle\langle q|+|1\rangle\langle m|;
\end{equation}
when $j\ge 2$, it is given by
\[
U_j = \sum_{q=1}^{m} |\pi_j(q)\rangle \langle q|,\quad\pi_j(q)=
\begin{cases}
q, & 1 \le q \le j - 1,\\[4pt]
q+1, & j \le q \le m-1,\\[4pt]
j, & q = m.
\end{cases}
\]
By uncomputing the register $C$, we obtain the state
\begin{equation}
    \left|\phi^{(j)}\right\rangle=\frac{1}{\sqrt{m-m_x}}\left(\sum_{q'=j+1-\sum_{k=1}^{j}\mathbf{1}_{x_k = x}}^{m-m_x}|q'+\sum_{k=1}^{j}\mathbf{1}_{x_k = x}\rangle_I+\sum_{k=1}^{j}(1-\mathbf{1}_{x_k = x})|k\rangle_I\right)|j\rangle_L\left|0\right\rangle_C.
\end{equation}
Consequently, by induction over $j$, the final state can be expressed as
\begin{equation}
	\begin{split}
		\left|\phi^{(m)}\right\rangle=\frac{1}{\sqrt{m-m_x}}\sum_{q\in\{k:x_k\ne x\}}\left|q\right\rangle_I|m\rangle_L|0\rangle_C.
	\end{split}
\end{equation}\\
\textbf{(c) Uncomputation operation:} This stage aims to eliminate the auxiliary registers, where register \(L\) can be further removed since \(m\) is known, yielding the target state
\begin{equation}
	\left|\psi\right\rangle=\frac{1}{\sqrt{m-m_x}}\sum_{q\in\{k:x_k\ne x\}}\left|q\right\rangle=\frac{1}{\sqrt{m-m_x}}\sum_{q:x_q\ne x}\left|q\right\rangle.
\end{equation}
At this stage, we have obtained the state $\left|\psi\right\rangle$, which encodes the distribution over the positions $q$ for which $x_q\ne x$.

\textbf{Space and pass complexity analysis.} We first analyze the space complexity. In the initialization stage, the registers $I$, $L$, and $C$ in the state $|\phi^{(0)}\rangle$ require $\log m$, $\log m$, and $O(1)$ qubits, respectively. To initialize register \(I\), the operator \(F_{m-m_x} \otimes \mathbb{I}_{m_x}\) is applied to an \(m\)-dimensional register initialized in the state \(|0\rangle\), which introduces additional ancillary space qubits. Following the construction in \cite{wiebe2015quantum}, preparing \(|\phi^{(0)}\rangle\) requires an additional \(\lceil \log(m - m_x + 1) \rceil + 2=O(\log m)\) qubits. 

In addition to this, constructing the indicator state $|\mathbf{1}_{x_j = x}\rangle$ requires additional workspace and coherent access to the data register $|x_j\rangle$ and the value $x$. We first apply a reversible bitwise XOR (implemented using CNOT gates)
\[
|x_j\rangle \mapsto |x_j\oplus x\rangle,
\]
which transforms the equality test into identifying the all-zero state. The result is then coherently encoded into a single ancillary qubit, yielding
\begin{equation}
|x_j\rangle|0\rangle\mapsto|x_j\rangle|\mathbf{1}_{x_j=x}\rangle.
\end{equation}
This construction requires only $O(\log n)$ ancillary qubits and can be implemented using standard reversible quantum circuits. 

During the subsequent update stage, only the registers $I$, $L$, and $C$ are maintained.  According to the definition of the operator $U_j$, its implementation involves piecewise permutations, which can be realized using reversible comparators together with controlled increment operations. Implementing these operations requires $O(\log m)$ ancillary workspace qubits, while all intermediate information is removed via uncomputation to maintain reversibility. The uncomputation procedure introduces no additional space overhead. Therefore, the QIPC subroutine requires $O(\log m+\log n)$ qubits of space.

In addition, the data stream is only accessed to generate the sequence of operations $Q^{(1)}\cdots Q^{(m)}$, which requires a single pass over the stream. The uncompute stage is performed by applying the inverse operations in reverse order. Since the stream length is known, it does not require any additional access to the data stream.

This completes the proof of Lemma \ref{thm:QIPC}.
\end{proof}

We are now ready to present a two-stage quantum streaming algorithm for entropy estimation. The main result is summarized in the following theorem.
\begin{theorem}\label{main result}
For parameters $\varepsilon,\delta >0$, there exists an $O(\log(1/\delta)\sqrt{\log m}/\varepsilon)$-pass quantum streaming algorithm that outputs an $(\varepsilon,\delta)$-approximation to the entropy of a data stream. The algorithm uses 
\[
O\!\left(\log m+\log n+\log\frac{1}{\varepsilon}\right)
\]
quantum and classical bits of space.
\end{theorem}
\begin{proof}
    To prove Theorem \ref{main result}, we first provide a detailed description of the two-stage quantum streaming algorithm for the streaming entropy estimation problem. We then explicitly characterize its space complexity and the number of passes over the data stream required by the algorithm.

\textbf{Two-stage quantum streaming algorithm.} Our algorithm proceeds in two stages: the first stage detects the presence of a majority element, determining whether any element $x$ has frequency $m_x > m/2$; the second stage performs entropy estimation, adapting to the outcome of the first stage depending on whether a majority element is present. A detailed implementation is provided in Algorithm \ref{entropy algorithm}, and the overall framework of the algorithm is illustrated in Fig. \ref{Fig3}.
\begin{algorithm}[h]
	\caption{Two-stage quantum streaming algorithm}
	\label{entropy algorithm}
	\begin{algorithmic}[1]
		\Statex \textbf{Input:} data stream $A=\langle x_1,x_2,\dots,x_m\rangle$ with $x_j\in [n]$, and parameters $0<\varepsilon,\delta<1$
		\Statex \textbf{Stage 1: majority element detection}
    	\State Apply the Boyer-Moore majority vote algorithm to identify a potential majority item $x$ with one pass over the stream.
		\State Make another pass over $A$ to determine the frequency $m_x$ of $x$.
		\Statex \textbf{Stage 2: entropy estimation}
		\If{$m_x\le m/2$}
        \State Initialize the state $|\Phi\rangle$ using the operator $F_m$.
		\State Invoke Lemma \ref{Expectation} to compute an approximation $\tilde{\mu}$ of $\mathbb{E}[X_q]$ given the oracle $O$.
		\State \textbf{Output:} $\tilde{\mu}$.
		
		\Else
        \State Prepare the state $|\psi\rangle$ using the proposed quantum subroutine QIPC. 
		\State Invoke Lemma \ref{Expectation} to compute an approximation $\tilde{\mu}'$ of $\mathbb{E}[X_q]^{'}$ given the oracle $O$.
		\State \textbf{Output:} $\frac{m-m_x}{m}\tilde{\mu}'+\frac{m_x}{m}\log\frac{m}{m_x}$.
		
		\EndIf
	\end{algorithmic} 
\end{algorithm}

\begin{figure*}[t]
	\centering
	\includegraphics[width=0.7\textwidth]{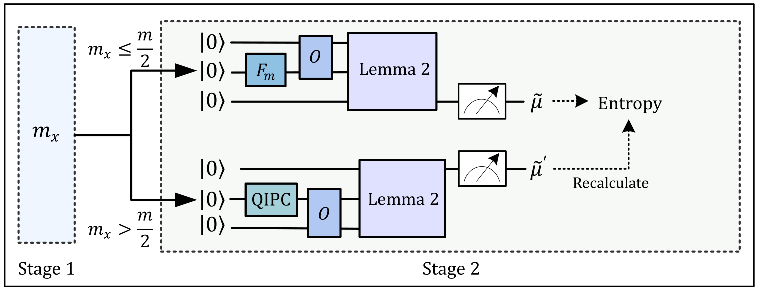}
	\caption{Overview of the proposed two-stage quantum streaming algorithm. Stage 1 determines the frequency of the majority element in the stream. Stage 2 uses $F_m$, which prepares an $m$-dimensional uniform superposition state, or QIPC, which constructs a superposition over the positions of all stream elements in $A$ except $x$. The oracle $O$ is implementable via the quantum streaming algorithm described in Lemma \ref{thm:QS-RVEE}. Lemma \ref{Expectation} is the proposed quantum query algorithm for estimating the expected value of a data stream random variable.}
	\label{Fig3}
\end{figure*}
\textbf{Stage 1: Majority element detection.} Detecting the majority element in a data stream is critical for our two-stage quantum streaming algorithm, determines the procedure for entropy estimation. We perform two passes over the data stream to verify whether the frequency of the majority element satisfies $m_x \le m/2$ or exceeds it ($m_x > m/2$). In the first pass, a candidate majority element $x$ is identified using the method of \cite{boyer1991mjrty}, and in the second pass, its frequency $m_x$ is counted. Regardless of whether the candidate $x$ is the true majority element, this procedure enables us to distinguish the two scenarios. Specifically, if \( m_x > m/2 \), then $x$ is indeed the majority element. Otherwise, any element \( x' \) with a frequency greater than or equal to \( x \), satisfies
\begin{equation*}
	m_x \leq m_{x{'}} \leq m/2,
\end{equation*}
implying that \( m_x \) suffices to differentiate the two cases.

\textbf{Stage 2: Entropy estimation.} After determining $m_x$, we classify entropy estimation methods into two cases. In the first case, there is no majority item, and estimating the entropy is relatively straightforward: our quantum query algorithm yields a reliable estimate. In the second case, a majority item is present, which can significantly affect the cost of entropy estimation. To address this, we temporarily remove the majority element and estimate the entropy of the remaining data. Once this estimate is obtained, the contribution of the removed element is added back to recover the total entropy.\\
\textbf{Case 1 ($m_x\le \frac{m}{2}$):}
In this case, the algorithm reaches line 6, and returns the estimate $\tilde{\mu}$. This estimate is obtained by applying Lemma \ref{Expectation} and serves as an $(\varepsilon,\delta)$-approximation of $\mathbb{E}[X_q]$. The lemma requires preparing a distribution over stream positions in order to invoke the oracle $O$, a prerequisite explicitly in the proof of Lemma \ref{Expectation}. In Case 1, this distribution is uniform over $q\in[m]$, which can be prepared by applying the operator $F_m$ to $|0\rangle$, yielding
\begin{equation}
    \left|\Phi\right\rangle=\frac{1}{\sqrt{m}}\sum_{q=1}^{m}|q\rangle.
\end{equation}
Here, $F_m$ denotes any transformation satisfying $F_m:|0\rangle\mapsto\frac{1}{\sqrt{m}}\sum_{i=1}^{m}|i\rangle$. We then apply Lemma \ref{Expectation} to estimate $\mathbb{E}[X_q]$ with failure probability at most $\delta$ and error at most $\varepsilon \mathbb{E}[X_q]$. According to \cite{amitSODA}, $X_q\in[-\log e,\log m]$, giving the parameters in Lemma \ref{Expectation} as $a=\log e$ and $b=\log m$, while $\mathbb{E}[X_q]\ge 1$. It follows that the number of oracle calls to $O$ is $O\left(\frac{2\log\left(1/\delta\right)\sqrt{\log em}}{\varepsilon}\right)$. Consequently, the estimate $\tilde{\mu}$ yields an $(\varepsilon,\delta)$-approximation of the entropy.\\
\textbf{Case 2 ($m_x> \frac{m}{2}$):} 
In this case, the algorithm reaches line 10, and outputs $\frac{m-m_x}{m}\tilde{\mu}'+\frac{m_x}{m}\log\frac{m}{m_x}$. When \(m_x > m/2\), we instead consider the random variable $X'_q$ defined on the substream \(A \backslash \{x\}\). We first invoke Lemma \ref{Expectation} to obtain an $(\varepsilon,\delta)$-approximation $\tilde{\mu}'$ of $\mathbb{E}[X'_q]$. Its application requires preparing the distribution over stream positions restricted to non-\(x\) elements, given by
\begin{equation}
	\left|\psi\right\rangle=\frac{1}{\sqrt{m-m_x}}\sum_{q:x_q\ne x}\left|q\right\rangle,
\end{equation}
which can be constructed using the Quantum Index-Position Conversion (QIPC) procedure in Lemma \ref{thm:QIPC}. We then apply Lemma \ref{Expectation} to estimate $\mathbb{E}[X'_q]$ with failure probability at most $\delta$ and error at most $\varepsilon \mathbb{E}[X'_q]$. Since $X'_q$ satisfies the same bounds as $X_q$ and also satisfies $\mathbb{E}[X'_q]\ge 1$, as established in \cite{amitSODA}, the number of oracle calls to \(O\) is $O\left(\frac{2\log\left(1/\delta\right)\sqrt{\log em}}{\varepsilon}\right)$. Finally, combining \(\tilde{\mu}'\) with the contribution of the majority element yields
\begin{equation}
	\frac{m-m_x}{m}\tilde{\mu}'+\frac{m_x}{m}\log\frac{m}{m_x}.
\end{equation}
Indeed, the error between the estimate and the entropy satisfies
	\begin{equation}
		\begin{split}
        \left|\frac{m-m_x}{m}\tilde{\mu}'+\frac{m_x}{m}\log\frac{m}{m_x}-H(p)\right|&=\left|\frac{m-m_x}{m}\tilde{\mu}'+\frac{m_x}{m}\log\frac{m}{m_x}-(\frac{m-m_x}{m}\mathbb{E}[X'_q]+\frac{m_x}{m}\log\frac{m}{m_x}\right)|\\
			&\le\frac{m-m_x}{m}|\tilde{\mu}'-\mathbb{E}[X'_q]|\\
			&\le \frac{m-m_x}{m}\varepsilon\mathbb{E}[X'_q]\\
			&\le\varepsilon\biggl(\frac{m-m_x}{m}\mathbb{E}[X'_q]+\frac{m_x}{m}\log\frac{m}{m_x}\biggr)\\
			&\le \varepsilon H(p).
		\end{split}
	\end{equation}
	Hence, the estimated entropy $\frac{m-m_x}{m}\tilde{\mu}'+\frac{m_x}{m}\log\frac{m}{m_x}$ is guaranteed to be an $(\varepsilon, \delta)$-approximation of the true entropy whenever $m_x > m/2$.

Finally, we analyze the space and pass complexities of the algorithm.

\textbf{Space and pass complexity analysis.}
	 We first analyze the space complexity of the proposed two-stage quantum streaming algorithm.
	
	In the first stage, the algorithm identifies a candidate majority element $x$ using $\log m+\log n$ bits, and computes its frequency using an additional $\log m$ bits.
	
	In the second stage, we distinguish two cases based on $m_x$. For $m_x \le m/2$, the quantum query algorithm uses $O\left(\log m+\log n+\log\frac{1}{\varepsilon}\right)$ qubits to output an $(\varepsilon,\delta)$-approximation of the entropy. The preparation of $|\Phi\rangle$ introduces an additional ancilla cost of $O(\log m)$ qubits (e.g., via \cite{wiebe2015quantum}), which is absorbed into the overall space complexity. For $m_x>m/2$, the state $|\psi\rangle$ is prepared using the QIPC subroutine, which requires $O(\log m+\log n)$ qubits by Lemma \ref{thm:QIPC}. Applying Lemma \ref{Expectation} yields an estimate $\tilde{\mu}'$ with the same space complexity as in the previous case. 
	
	In summary, the algorithm requires 
	\begin{equation*}
		O\left(\log m+\log n+\log\frac{1}{\varepsilon}\right)
	\end{equation*}
	bits and qubits of space. 

	We now analyze the pass complexity of the proposed algorithm.
    
    As described in Algorithm \ref{entropy algorithm}, the first stage makes two passes over the stream to determine the most frequent element $x$ and its frequency $m_x$. The first pass identifies a candidate element, while the second computes the frequency of $x$. 
	
	In the second stage, when \( m_x \le m/2 \), Lemma \ref{Expectation} makes $O\left(\frac{2\log(1/\delta)\sqrt{\log em}}{\varepsilon}\right)$ oracle calls, which translates into $O\left(\frac{4\log(1/\delta)\sqrt{\log em}}{\varepsilon}\right)$ stream passes. For $m_x > m/2$, Lemma \ref{Expectation} requires $O\left(\frac{2\log(1/\delta)\sqrt{\log em}}{\varepsilon}\right)$ calls to the oracle $O$ together with the quantum subroutine QIPC. This yields  $O\left(\frac{6\log(1/\delta)\sqrt{\log em}}{\varepsilon}\right)$ stream passes, where for each call to $O$ requires two passes and QIPC requires one pass.
	
    In both cases, the total number of passes over the data stream is bounded by
	\begin{equation*}
		O\left(\frac{\log(1/\delta)\sqrt{\log m}}{\varepsilon}\right).
	\end{equation*}
    
   This completes the proof of Theorem \ref{main result}.
\end{proof}

\subsection{Classical space lower bounds for streaming entropy estimation}\label{Classical lower bound}

We generalize the classical space lower bound for one-pass streaming entropy estimation in \cite{amitSODA} to the multi-pass setting. More specifically, we establish a classical space lower bound for multi-pass streaming entropy estimation via a reduction from the Gap Hamming Distance (GHD) problem. The corresponding result is summarized in the following theorem.

\begin{theorem}[Lower bound on streaming entropy estimation] \label{lowerbound}
Any randomized $T$-pass streaming algorithm that outputs an $(\varepsilon,\delta)$-approximation of the Shannon entropy must use
\[
\Omega\!\left(\frac{1}{T \varepsilon^2 \log^2(1/\varepsilon)}\right)
\]
bits of space.
\end{theorem}
\begin{proof}
We formally define the GHD problem. 
Two parties, Alice and Bob, each hold binary strings $\mathbf{x}=\{x_j\}_{j=1}^{m}, \mathbf{y}=\{y_j\}_{j=1}^{m} \in \{0,1\}^m$, respectively, and must decide whether the Hamming distance $d(\mathbf{x}, \mathbf{y})$ satisfies one of the following conditions:
	\begin{equation*}
		d(\mathbf{x},\mathbf{y}) \le m/2-\sqrt{m} \quad\text{or}\quad d(\mathbf{x},\mathbf{y}) \ge m/2 + \sqrt{m}.
	\end{equation*}
	The goal is to determine the value of
	\begin{equation}
		\text{GHD}(\mathbf{x},\mathbf{y})=\begin{cases}
			\text{near},& \text{if $d(\mathbf{x},\mathbf{y}) \le m/2-\sqrt{m}$}, \\
			\text{far},& \text{if $d(\mathbf{x},\mathbf{y}) \ge m/2 + \sqrt{m}$}.
		\end{cases}
	\end{equation}
	It is known that any randomized protocol for the GHD problem requires $\Omega(m)$ bits of communication \cite{chakrabarti2011optimal}.

Based on the definition of the GHD problem, we now present a reduction from GHD to entropy estimation. The reduction proceeds in three steps. First, the input of the GHD problem is mapped to an input stream for entropy estimation. Second, we show that the estimated entropy from the protocol can be used to solve GHD by selecting an appropriate error parameter \(\varepsilon\). Third, a multi-pass streaming algorithm is employed to construct a multi-round communication protocol. We now detail the reduction process.
    
    First, Alice and Bob jointly construct the stream $A$ for the entropy estimation problem from their respective inputs to the GHD problem. Specifically, each party constructs the substreams $A_{\mathbf{x}} = \{ j+mx_j\}_{j=1}^{m}$ and $A_{\mathbf{y}} = \{ j+m y_j\}_{j=1}^{m}$, respectively. The full stream is then given by $A =A_{\mathbf{x}} \circ A_{\mathbf{y}}$, with total length $2m$. 
    
	Next, we analyze how the estimated entropy $\tilde{H}(p)$ can be used to distinguish whether $\text{GHD}(\mathbf{x},\mathbf{y})$ is in the ``near'' or ``far'' case. From the definitions of entropy and GHD, the entropy of the concatenated stream \(A\) is
	\begin{equation}
		\begin{aligned}
			H(p)&=\frac{2d(\mathbf{x},\mathbf{y})}{2m}\log 2m+\left(m-d(\mathbf{x},\mathbf{y})\right)\frac{2}{2m}\log\frac{2m}{2}=\log m+\frac{d(\mathbf{x},\mathbf{y})}{m}.
		\end{aligned}
	\end{equation}
	Here, $m-d(\mathbf{x},\mathbf{y})$ elements appear twice with probability $1/m$, while $2d(\mathbf{x},\mathbf{y})$ elements appear once with probability $1/(2m)$. To distinguish the ``near'' and ``far'', it suffices to ensure that
	\begin{align}
		(1+\varepsilon)
		&\max_{\text{GHD}(\mathbf{x},\mathbf{y})=\text{near}} H(p) \notag< (1-\varepsilon)
		\min_{\text{GHD}(\mathbf{x},\mathbf{y})=\text{far}} H(p),
	\end{align}
	where the maximum and minimum are attained at $d(\mathbf{x},\mathbf{y})=\frac{m}{2}-\sqrt{m}$ and $d(\mathbf{x},\mathbf{y})=\frac{m}{2}+\sqrt{m}$, respectively. In the ``near'' case, we obtain the upper bound
	\begin{equation}\label{near_bound}
		\max_{\text{GHD}(\mathbf{x},\mathbf{y})=\text{near}}\tilde{H}(p)=(1+\varepsilon)\max_{\text{GHD}(\mathbf{x},\mathbf{y})=\text{near}}H(p)\le\left(1+\varepsilon)(\log m+\frac{m/2-\sqrt{m}}{m}\right),
	\end{equation}
	while in the ``far'' case, the lower bound is
	\begin{equation}\label{far_bound}
		\min_{\text{GHD}(\mathbf{x},\mathbf{y})=\text{far}}\tilde{H}(p)=(1-\varepsilon)
		\min_{\text{GHD}(\mathbf{x},\mathbf{y})=\text{far}} H(p)\ge\left(1-\varepsilon)(\log m+\frac{m/2+\sqrt{m}}{m}\right).
	\end{equation}
	Combining Eqs. (\ref{near_bound}, \ref{far_bound}), We obtain $\varepsilon<\frac{1}{\sqrt{m}(\log m+1/2)}$,which leads to the decision rule:
	\begin{equation}
		\begin{split}
			\text{GHD}&(\mathbf{x}, \mathbf{y})=\begin{cases}
				\text{near}&\text{if\quad $\tilde{H}(p)\le \log m+1/2-\frac{1}{\sqrt{m}}$}, \\
				\text{far}&\text{if\quad $\tilde{H}(p) \geq \log m+1/2+\frac{1}{\sqrt{m}}$}.
			\end{cases}
		\end{split} 
	\end{equation}
	Therefore, any algorithm $\mathcal{C}$ produces an $(\varepsilon,\delta)$-approximation of $H(p)$ with $\varepsilon<\frac{1}{\sqrt{m}(\log m+1/2)}$ suffices to solve the GHD problem.
    
We now use a multi-pass streaming entropy algorithm to construct a multi-round communication protocol. Let \(\mathcal{C}\) be a \(T\)-pass streaming algorithm that provides an \((\varepsilon, \delta)\)-approximation of the true entropy \(H(p)\) using at most \(S\) bits of memory, and let \(\tilde{H}(p)\) denote its output. For the constructed data stream $A=A_{\mathbf{x}}\circ A_{\mathbf{y}}$, running $\mathcal{C}$ induces a two-party randomized communication protocol with $2T-1$ rounds. 
\begin{figure*}[t]
	\centering
	\includegraphics[width=0.85\textwidth]{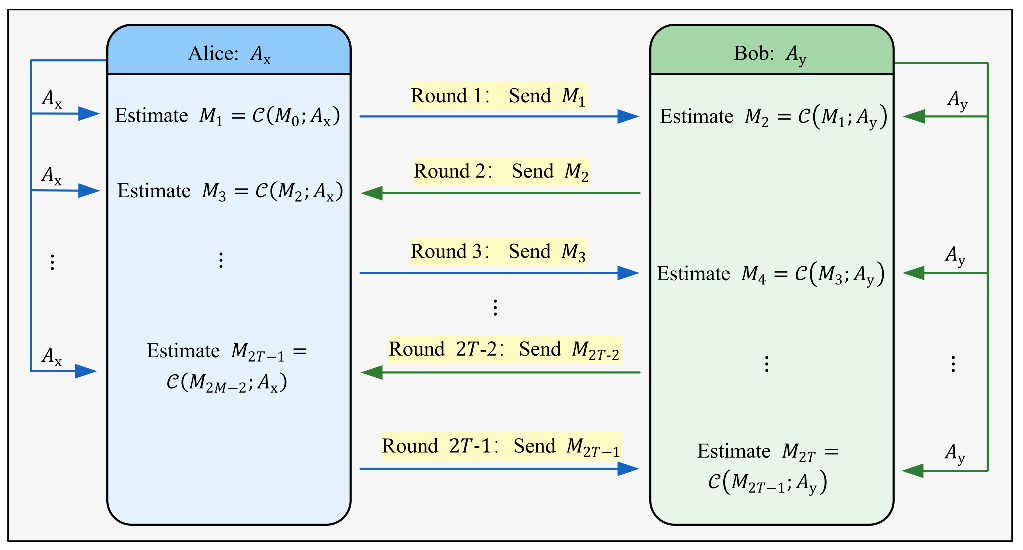}
	\caption{Schematic diagram of the $(2T-1)$-round communication protocol derived from the streaming algorithm $\mathcal{C}$. Here, \(M_0\) denotes an empty message representing the initial stage of the protocol. Alice first processes the stream block \(A_{\mathrm{x}}\) using algorithm \(\mathcal{C}\), starting from \(M_0\), and sends the resulting memory state \(M_1\) to Bob. Bob then resumes the computation from the received state \(M_1\), processes \(A_{\mathrm{y}}\) using \(\mathcal{C}\), and obtains an updated memory state \(M_2\), which is sent back to Alice. Alice subsequently continues the computation from \(M_2\), processes \(A_{\mathrm{x}}\), and sends the resulting state \(M_3\) to Bob. This procedure continues alternately, with each party resuming the computation from the memory state received from the other. More generally, a \(T\)-pass streaming algorithm naturally induces a \((2T-1)\)-round communication protocol, where each transmitted message corresponds to the memory state maintained by the streaming algorithm. After the final transmission, Bob resumes the computation from the received state and processes \(A_{\mathrm{y}}\) to produce the output $M_{2T}$.}
	\label{Fig4}
\end{figure*}
\paragraph{$(2T-1)$-round communication protocol}
For notational convenience, let \(\mathcal{C}(M_{2i-2},A_{\mathbf{x}})\) the memory state obtained by continuing the execution of the streaming algorithm $\mathcal{C}$ from the memory state $M_{2i-2}$ and processing the stream block \(A_{\mathbf{x}}\). Analogously, \(\mathcal{C}(M_{2i-1},A_{\mathbf{y}})\) denotes the memory state obtained by continuing \(\mathcal{C}\) from the state $M_{2i-1}$ and processing the stream block \(A_{\mathbf{y}}\). Fig. \ref{Fig4} illustrates this transformation.

For \(i=1,2,\dots,T\) (the \(i\)-th pass):
\begin{enumerate}[itemsep=2pt, topsep=2pt]
\item \textbf{Alice processes \(A_{\mathbf{x}}\):}
\begin{itemize}
\item Alice resumes algorithm \(\mathcal{C}\) from memory state \(M_{2i-2}\), where \(M_0\) denotes an empty message, and processes the local stream block \(A_{\mathbf{x}}\), resulting in
\[
M_{2i-1}
=
\mathcal{C}(M_{2i-2},A_{\mathbf{x}}).
\]

\item Alice sends \(M_{2i-1}\) to Bob (communication round \(2i-1\)).
\end{itemize}

\item \textbf{Bob processes \(A_{\mathbf{y}}\):}
\begin{itemize}
\item Bob resumes algorithm \(\mathcal{C}\) from the received memory state \(M_{2i-1}\) and processes the local stream block \(A_{\mathbf{y}}\), resulting in
\[
M_{2i}
=
\mathcal{C}(M_{2i-1},A_{\mathbf{y}}).
\]

\item If \(i<T\), Bob sends \(M_{2i}\) back to Alice (communication round \(2i\)); otherwise (\(i=T\)), Bob outputs the final answer based on \(M_{2T}\).
\end{itemize}
\end{enumerate}
In the constructed protocol, each transmitted message has size at most $S$ bits, and the protocol consists of $2T-1$ rounds. Hence, the total communication complexity is bounded by $(2T-1)S$. Since the constructed protocol solves the GHD problem, the communication lower bound for GHD implies that
\begin{equation}
    (2T-1)S\ge\Omega(m).
\end{equation}
Using $\varepsilon<\frac{1}{\sqrt{m}(\log m+1/2)}$ and choosing $\varepsilon\approx \frac{1}{\sqrt{m}\log m}$, we obtain the space lower bound for any multi-pass streaming algorithm:
\begin{equation}
    S=\Omega\left(\frac{1}{T\varepsilon^2\log^2 \varepsilon^{-1}}\right).
\end{equation}
	This establishes the claimed lower bound for any $T$-pass streaming algorithm approximating the entropy.
\end{proof}

\section*{Data availability}
No datasets were generated or analyzed during the current study.
\section*{Code availability}
No Code were generated or analyzed during the current study.

\backmatter

\bmhead{Supplementary information}
The online version contains
supplementary material available at ...

\bmhead{Acknowledgements}
This work was supported by the National Key Research and Development Program of China (Grant No. 2024YFB4504004), National Natural Science Foundation of China (Grants No. 62171131, 92465202, 62272492), Fujian Province Natural Science Foundation (Grant No. 2022J01186 and 2023J01533), Fujian Province Young and Middle-aged Teacher Education Research Project (Grant No. JAT231018) and the Guangdong Provincial Quantum Science Strategic Initiative (Grant No. GDZX2303007, GDZX2503001).

\bibliography{sn-bibliography}

\end{document}